\documentclass[12pt]{article}  
\usepackage[english]{babel}
\usepackage[toc,page]{appendix}
\usepackage{bbm}
\usepackage{fullpage}
\usepackage[top=1in, bottom=1.25in, left=1in, right=1in]{geometry}
\usepackage{etaremune}
\usepackage{enumitem}

\usepackage{amssymb,amsmath,amsthm}

\usepackage{hyperref}

\usepackage{xcolor}

\DeclareSymbolFont{rsfs}{U}{rsfs}{m}{n}
\DeclareSymbolFontAlphabet{\mathscrsfs}{rsfs}
\usepackage{mathrsfs}

\usepackage{url}

\hypersetup{
  colorlinks   = true, 
  urlcolor     = blue, 
  linkcolor    = blue, 
  citecolor   = red 
}

\usepackage[labelformat=empty]{caption}

\newtheorem{theorem}{Theorem}[section]
\newtheorem{lemma}[theorem]{Lemma}

\newtheorem{proposition}[theorem]{Proposition}
\newtheorem{corollary}[theorem]{Corollary}

\theoremstyle{definition}
\newtheorem{definition}{Definition}
\newtheorem{remark}[theorem]{Remark}

\numberwithin{equation}{section}

\usepackage{times}

\def\iid{{\text{i.i.d.~}}}

\newcommand{\bea}{\begin{eqnarray}}
\newcommand{\eea}{\end{eqnarray}}
\newcommand{\<}{\langle}
\renewcommand{\>}{\rangle}

\newcommand{\wt}{\widetilde}

\newcommand{\wh}{\widehat}

\def\poly{{\rm poly}}

\def\eps{{\varepsilon}}

\def\supp{{\rm supp}}

\def\<{\langle}
\def\>{\rangle}

\def\Geom{{\rm Geom}}

\def\cM{{\cal M}}
\def\cN{{\cal N}}

\def\tot{{\rm tot}}

\def\b0{{\boldsymbol{0}}}

\def\Var{{\rm Var}}

\def\cA{{\mathcal A}}
\def\cI{{\mathcal I}}

\DeclareMathOperator*{\argmax}{argmax}

\renewcommand{\b}{\mathbf{b}}

\def\lt{\left}
\def\rt{\right}

\def\la{\langle}
\def\ra{\rangle}

\def\eps{\varepsilon}

\def\bbC{{\mathbb{C}}}
\def\bbE{{\mathbb{E}}}

\def\bbN{{\mathbb{N}}}
\def\bbP{{\mathbb{P}}}
\def\bbR{{\mathbb{R}}}

\def\bbZ{{\mathbb{Z}}}

\def\cN{{\mathcal{N}}}

\title{Tight Space Lower Bound for Pseudo-Deterministic Approximate Counting}
\author{Ofer Grossman \and Meghal Gupta \and Mark Sellke}
\date{}

\begin{document}

\maketitle

\begin{abstract}
We investigate one of the most basic problems in streaming algorithms: approximating the number of elements in the stream.
Famously, \cite{morris1978counting} gave a randomized algorithm achieving a constant-factor approximation error for streams of length at most $N$ in space $O(\log\log N)$. We investigate the pseudo-deterministic complexity of the problem and prove a tight $\Omega(\log{N})$ lower bound, thus resolving a problem of \cite{goldwasser2019pseudo}.
\end{abstract}

\section{Introduction}

The study of streaming algorithms originated with the seminal paper of Morris \cite{morris1978counting}, which
gave a low-memory randomized algorithm to approximately count a number of elements which arrive online in a stream.
Roughly speaking, the idea is to have a counter which approximates $\log n$, where $n$ is the number of elements seen in the stream so far. Each time the algorithm encounters an element from the stream, it increases the counter with probability about $\frac{1}{n}$.
As later proved by \cite{flajolet1985approximate}, Morris's algorithm achieves a constant-factor approximation error for streams of length at most $N$ in space $O(\log\log N)$.

Morris's algorithm has the property that running it multiple times on the same stream may result in different approximations. That is, if Alice runs the algorithm on the same stream as Bob (but using different randomness), Alice may get some approximation (such as $2^{30}$), and Bob (running the same algorithm but with independent randomness) may get a different approximation (such as $2^{29}$). 
Is this behavior inherent? That is, could there exist a low-space algorithm which, while being randomized, for all streams with high probability both Alice and Bob will end up with the \textit{same} approximation for the length?
Such algorithms, namely those which output the same output with high probability when run multiple times on the same input, are called \textit{pseudo-deterministic}. 



This question of the pseudo-deterministic space complexity of approximate counting in a stream was first posed in \cite{goldwasser2019pseudo}. 
Our main result fully resolves the problem by giving a tight $\Omega(\log N)$ lower bound.
In fact, our lower bound applies for an easier threshold version of the problem asking to distinguish between streams of length at most $N$ versus at least $M$ for any integers $N<M$. Moreover, it depends only on $N$, i.e. the problem is hard even when $M$ is arbitrarily larger than $N$.

\begin{theorem}[Informal]

For any $N<M$, a pseudo-deterministic streaming algorithm to distinguish between streams of length at most $N$ and at least $M$ must use $\Omega(\log N)$ space.
\end{theorem}

\noindent Concurrently, \cite{braverman2023lower} showed a non-trivial $\Omega(\sqrt{\log n /\log \log n})$ space lower bound for this problem using different techniques.


\paragraph{Motivation for Pseudo-Determinism.} 
Randomization is a powerful tool in algorithm design.
For example, almost all classic streaming algorithms (such as for heavy hitters, approximate counting, $\ell_p$ approximation, finding a nonzero entry in a vector for turnstile algorithms, counting the number of distinct elements in a stream) cannot be derandomized without incurring a dramatic blowup in space complexity. 
However randomized algorithms also have several downsides compared to deterministic algorithms:
\begin{itemize}
\item They have a nonzero probability of error or failure.
\item Random bits must be sampled, which may be infeasible or computationally expensive.
\item The output is not predictable. That is, multiple executions of the algorithm on the same input may result in different outputs.
\end{itemize}

The first two points above have both significantly informed research on randomized algorithms. For example, the first point motivates the distinction between Monte Carlo and Las Vegas algorithms as well as RP versus BPP, and the need for improvements over naive amplification \cite{impagliazzo1989recycle,grossman2016amplification}. The second point has led to work on pseudorandom generators \cite{haastad1999pseudorandom}, how to recycle random bits \cite{impagliazzo1989recycle}, and efficient extractors \cite{trevisan2001extractors}, to mention a few. The study of \emph{pseudo-determinism} as introduced by Gat and Goldwasser \cite{GG} seeks to address the third problem. 
Our result shows that the third downside cannot be avoided by any algorithm that solves the approximate counting (or threshold) problem in low space.




\subsection{Related Work}

\paragraph{Prior Work on Approximate Counting in a Stream.} 
The study of streaming algorithms began with the work Morris \cite{morris1978counting} on approximate counting; a rigorous analysis was given later in \cite{flajolet1985approximate}. As explained earlier, the (randomized) Morris counter requires logarithmically fewer states than a (deterministic) exact counter.
The approximate counter has been useful as a theoretical primitive for other streaming algorithms \cite{ajtai2002approximate,gronemeier2009applying,kane2010exact,bhattacharyya2018optimal,jayaram2019towards} and its performance has been evaluated extensively
\cite{cvetkovski2007algorithm,csuros2010approximate,xu2021memory}.

The optimal dependence on the error level and probability in approximate counting was determined recently in \cite{nelson2022optimal}. Moreover, \cite{aden2022amortized} studied the amortized complexity of maintaining several approximate counters rather than just one.

\paragraph{Prior Work on Pseudo-Determinism.}
Pseudo-deterministic algorithms were introduced by Gat and Goldwasser \cite{GG} who were motivated by applications to cryptography. Since then, such algorithms have since been studied in the context of standard (sequential algorithms) \cite{roots, OS}, average case algorithms \cite{dhiraj}, parallel algorithms \cite{matching, ghosh2021matroid}, decision tree algorithms \cite{GGR, k-pseudodeterminism, dixon2021complete, goldwasser2021pseudo}, interactive proofs \cite{proofs, goemans2019doubly}, learning algorithms \cite{OS2}, approximation algorithms \cite{OS2, dixon}, low space algorithms \cite{reproducibility}, complexity theory \cite{lu2021pseudodeterministic, dixon2022pseudodeterminism} and more. 
In several of these settings, pseudo-deterministic algorithms were shown to outperform deterministic algorithms and sometimes even match the state of the art performance for randomized algorithms.

The paper \cite{goldwasser2019pseudo} initiated the study of pseudo-deterministic streaming algorithms. They examined various problems (such as finding heavy hitters, $\ell_2$ approximation, and finding a nonzero entry in a vector for turnstile algorithms) and asked whether approximate counting admits a pseudo-deterministic algorithm using $O(\log \log n)$ bits of space. We give an optimal $\Omega(\log n)$ lower bound.

As this work was being finalized, \cite{braverman2023lower} independently showed a $\Omega(\sqrt{\log n /\log \log n})$ lower bound.
Their proof proceeds via reduction to one-way communication complexity. Our proof, on the other hand, takes a completely different approach and analyzes the complexity directly by modeling the streaming algorithm as a Markov chain and showing it behaves as an ensemble of cyclic parts in a suitable sense.

\subsection{Main Result}

We consider the problem of pseudo-deterministic approximate counting in a stream. We first recall the definition of a pseudo-deterministic streaming algorithm:

\begin{definition} 
\label{def:PD-streaming}
We define a streaming algorithm $\cA$ to be \textit{pseudo-deterministic} if there exists some function $F$ such that for all valid input streams $x$,
\[
\bbP_r[\cA(x, r) = F(x)] \ge 2/3
\]
where $r$ is the randomness sampled and used by the algorithm $\cA$.
\end{definition}

Our main result gives a tight bound of $\Omega(\log N)$ bits of memory. 
In fact, we prove that pseudo-deterministically distinguishing between streams of length $T\leq N$ and $T\geq f(N)$ requires $\Omega(\log N)$ bits of memory for \emph{any} function $f:\bbN\to\bbN$ where $f(N) > N$.

\begin{definition}[Pseudo-Deterministic Approximate Threshold Problem]
\label{prob:main}
    Fix $N, M \in \bbN$. Suppose $\cA$ is a pseudo-deterministic streaming algorithm such that with $T$ the stream length:
    \begin{enumerate}
        \item $F(T)=1$ for $T\leq N$,
        \item $F(T)=0$ for $T\geq M$.
    \end{enumerate}
    Then we say $\cA$ solves the $(N,M)$-pseudo-deterministic approximate threshold problem. 
\end{definition}

\begin{theorem}
\label{thm:main}
    For any $f:\bbN\to\bbN$ where $f(N) > N$, any sequence of pseudo-deterministic algorithms $\cA_N$ which solve the $(N,f(N))$-approximate threshold problem requires $\Omega(\log N)$ bits of space.
\end{theorem}

Since counting is harder than thresholding, the next corollary is an immediate consequence.

\begin{corollary}
\label{cor:main}
    Any pseudo-deterministic streaming algorithm which solves approximate counting up to a constant multiplicative factor for stream lengths at most $N$ requires $\Omega(\log N)$ bits of space.
\end{corollary}

\begin{remark}
    Our result also implies a space lower bound for solving a variant of the heavy hitters problem pseudo-deterministically.
    Consider a version of this problem where given a $\{0,1\}$-valued stream, we aim to output a bit that appeared at least $10\%$ of the time.
    It follows from Corollary~\ref{cor:main} 
    that any pseudo-deterministic streaming algorithm solving this problem requires $\Omega(\log N)$ bits of space.
    Indeed, if it is public information that the first $9N$ bits are $0$ and the rest are $1$, then we are reduced to the $(N,81N)$-approximate threshold problem by counting the $1$'s. 
    Randomized algorithms again can solve the problem with $O(\log\log N)$ bits as it suffices to maintain a pair of Morris counters for the $0$'s and $1$'s separately.
    (Note that while heavy hitters is often considered to be solvable deterministically in constant space, such solutions typically assume a model where exact counting uses constant space.)
\end{remark}

\subsection{Markov Chain Formulation}
\label{subsec:markov-chain-setup}

It will be helpful to reframe the approximate threshold problem and pseudo-determinism in the language of Markov chains.
Note that any $b$-bit pseudo-deterministic streaming algorithm can be described as a Markov chain on $2^b$ states. 
Throughout the rest of the paper, we use this formulation instead, which we make precise below.

Let $\cM_n$ be an arbitrary Markov chain on state space $V=\{v_1,\dots,v_n\}$ with starting state $x_0=v_1$, and let $x_1,x_2,\dots$ be the random states at each subsequent time. 
Moreover, let $U\subseteq V$ be a distinguished subset of $V$.

\begin{definition}[{Markov Chain Solution to Pseudo-Deterministic Approximate Thresholding}]
~\\
We say $(\cM_n,U)$ is a $(N,f(N))$-solution to the approximate threshold problem if:
\begin{enumerate}[label=(\roman*)]
    \item $\bbP[x_t\in U]\geq 2/3$ for all $t\leq N$,
    \item $\bbP[x_t\in U]\leq 1/3$ for all $t\geq f(N)$.
\end{enumerate}
We say $\cM_n$ is pseudo-deterministic if $\bbP[x_t\in U]\in [0,1/3]\cup [2/3,1]$ for all $t$.
\end{definition}

We will show the following:

\begin{theorem}
\label{thm:main2}
    For any pseudo-deterministic $(\cM_n,U)$ that $(N,f(N))$-solves the approximate threshold problem for some $f(N) > N$, it holds that $n\geq N^{\Omega(1)}$.
\end{theorem}

It is easy to see that Theorem~\ref{thm:main2} implies Theorem~\ref{thm:main} by viewing a streaming algorithm using $b$ bits of memory as a Markov chain on $2^b$ states. 
Here $U$ is the set of algorithm states on which the algorithm outputs $1$.


\section{Technical Overview}

\subsection{Illustrative Examples}

We begin by describing two examples of Markov chains that fail to solve the pseudo-deterministic approximate threshold problem, but end up being surprisingly illustrative of the general case.

\paragraph{Example 1: Threshold Morris Counter.}
We first describe how to solve approximate thresholding in the usual randomized (non-pseudo-deterministic) setting, where $M=10N$ and for $t\in [N,10N]$, it is \emph{not} required that $\bbP[x_t\in U] \in [0,1/3] \cup [2/3,1]$.
We will use a version of a Morris counter with a simple two-state implementation (since we aim to threshold rather than count).

Our Markov chain will have states $(v_1,v_2)$ where $v_2$ is terminal, and the transition from $v_1$ to $v_2$ occurs with probability $\frac{1}{3N}$. Formally, $\bbP[x_{t+1}=v_2~|~x_t=v_2]=1$, and $\bbP[x_{t+1}=v_2~|~x_t=v_1]=\frac{1}{3N}$.
Setting $U=\{v_1\}$, this chain gives a randomized algorithm for $(N,10N)$ approximate thresholding. However, this example is not pseudo-deterministic because the function $t\mapsto\bbP[x_t\in U]$ decays gradually. 
In particular it is easy to see that $\bbP[x_{t+1}\in U]\geq 0.99\cdot \bbP[x_t\in U]$, so by the discrete-time intermediate value theorem there exists $t$ such that $\bbP[x_t\in U]\notin [0,1/3]\cup [2/3,1]$.

\paragraph{Example 2: Prime Cycles.}
In our second example, we discuss a futile attempt to solve a simpler version of the approximate threshold problem where it is given that the stream's length is at most $10M=100N$.
Informally speaking, on the first step the Markov chain picks a prime $p_i$. On all subsequent steps, it simply records the current time modulo $p_i$.
A main idea in our proof, outlined in the next subsection, will be to show that any Markov chain $\cM_n$ behaves like (a more complicated version of) this example.

Let $k=[n^{1/4},2n^{1/4}]$ and choose $k$ distinct primes $p_1,\dots,p_k\in [n^{1/3},2n^{1/3}]$.
Noting that $\sum_{i=1}^k p_i\leq n$, we can construct a Markov chain $\cM_n$ which contains an initial state $v_0$ as well as a deterministic $p_i$-cycle for each $1\leq i\leq k$. Here a $p$-cycle consists of vertices $(v^{(p)}_0,\dots,v^{(p)}_{p-1})$ with dynamics deterministically incrementing the subscript by $1$ modulo $p$ each time. 
At the first timestep, $\cM_n$ moves to $x_1=v^{(p_i)}_1$ for a uniformly random $p_i$.
One could try to solve the pseudo-deterministic approximate threshold problem (with restricted stream length at most $100N$) by choosing a special subset $U\subseteq \bigcup_{1\leq i\leq k} \{v^{(p_i)}_0,\dots,v^{(p_i)}_{p_i-1}\}$ of the state space of $\cM_n$.

We argue below that this is not possible, which will serve as a useful warmup for the main proof.
Note that by the discrete-time intermediate value theorem, for any such solution there must exist $0\leq T\leq 9N$ so that $\bbP[x_t\in U]\geq 2/3$ holds for roughly half of the values $t\in [T,T+1,\dots,T+N]$. 
To contradict pseudo-determinism, it suffices to show an upper bound of $o(1)$ on the variance of $\bbP[x_t\in U]$ over the range $t\in [T,T+1,\dots,T+N]$. To do this, it suffices to show that for some constant $C$,
\begin{equation}
\label{eq:sum-of-squares-want}
    \sum_{t\in [T,T+1,\dots,T+N]} (\bbP[x_t\in U]-C)^2\leq o(N).
\end{equation}
We will set $C$ to be the average value we would expect $\bbP[x_t\in U]$ to take on, which is $C=\bbE_{i\in [k]} F(p_i)$, where $F(p_i)$ is the fraction of the $\{v^{(p_i)}_0,\dots,v^{(p_i)}_{p_i-1}\}$ that are in $U$. Rewriting the summand on the left-hand side, we need to bound
\begin{align*}
    \sum_{t\in [T,\dots,T+N]} (\bbP[x_t\in U]-C)^2 = \sum_{t\in [T,\dots,T+N]} 
    \left( 
    \frac{1}{k}
    \sum_{1\leq i \leq k}\left[ \mathbbm{1}( v^{(p_i)}_{t}\in U) - F(p_i) \right] \right)^2.
\end{align*}
where the subscript of $v^{(p_i)}$ is viewed mod $p_i$ and $\mathbbm{1}$ represents the indicator function. Expanding this expression and switching the order of summation, it is equal to
\begin{align}
\notag
    &\frac{1}{k^2}
    \sum_{1\leq i \leq k} \left[ \sum_{t\in [T,\dots,T+N]} \left(\mathbbm{1}( v^{(p_i)}_{t}\in U) - F(p_i)\right)^2 \right] 
    \\
\label{eq:summand-above}
    & 
    +
    \frac{1}{k^2}
    \sum_{1\leq i,j \leq k} \left[ \sum_{t\in [T,\dots,T+N]} \left(\mathbbm{1}( v^{(p_i)}_{t}\in U) - F(p_i)\right)\cdot \left(\mathbbm{1}( v^{(p_j)}_{t}\in U) - F(p_j)\right) \right]
\end{align}
The first term is $O(N/k) = o(N)$, so we only need to bound the second term; we fix $i,j$ and uniformly estimate the inner expression by $o(N)$. For each prime $p$, the expression $\mathbbm{1}( v^{(p)}_{t}\in U) - F(p)$ averages to $0$ over any interval of length $p$. Thus by the Chinese remainder theorem, the summand in \eqref{eq:summand-above} has average value $0$ on every length $p_ip_j$ interval. After accounting for this cancellation, the number of remaining terms is at most $p_i p_j\leq O(n^{2/3})\leq o(N)$. Since the summands have absolute value at most $1$, combining yields the bound \eqref{eq:sum-of-squares-want}, thus contradicting pseudo-determinism.

\subsection{Proof Outline}

Suppose for sake of contradiction that the Markov chain $(\cM_n, U)$ is pseudo-deterministic and solves the approximate threshold problem.
The main idea of our proof is to extract from $\cM_n$ a subsystem behaving roughly like Example $2$ above, in the sense that it is an ensemble of cycles with different period lengths. 
We leverage this behavior around a time $T$ where $\bbP[x_s\in U]\geq 2/3$ holds for roughly half of the values $s\in [T,T+\poly(n)]$, which exists by the discrete-time intermediate value theorem.
Finally, we use Fourier analysis to prove there exists some $s\in [T,T+\poly(n)]$ where $\bbP[x_s\in U]$ is not too biased, contradicting the pseudo-determinism of $(\cM_n,U)$.

\paragraph{Recurrent Behavior on Moderate Time-Scales.}
The first step is to bypass the issue of permanent irreversible transitions as in Example $1$ by considering a random time. We choose a random time $1\leq t\leq \poly(n)$ and show that with high probability the chain \emph{behaves recurrently} at vertex $x_t$. 
More precisely, let $t+r$ be the next time at which $x_{t+r}=x_t$ (where $r=\infty$ if the chain never returns).
We show in Lemma~\ref{lem:find-t} that with high probability,
\begin{align}
\label{eq:short-cycles-overview}
    \bbP[r\leq 10n]&\geq 1/2,
    \\
\label{eq:long-cycles-overview}
    \bbP[r\leq n^{18}]&\geq 1-n^{-16}.
\end{align}

Using this guarantee, we can view the Markov chain from the perspective of the (random) state $x_t$ by considering the sequence of cycle lengths. The bound \eqref{eq:long-cycles-overview} allows us to define this process up to large $\poly(n)$ number of time steps, with all cycle lengths at most $n^{18}$ with high probability. This circumvents the behavior in Example $1$: although the chain might eventually transition to a terminal state and never return to $x_t$, it tends not to do so during a certain $\poly(n)$ time window.

\paragraph{Decomposition into Periodic Parts.}
Next, using \eqref{eq:short-cycles-overview}, it follows that there exists an integer $1\leq m_v\leq 10n$ such that conditioned on $x_t=v$, each cycle returning to $v$ has probability $\Omega(1/n)$ to have length exactly $m_v$.
We call cycles of length exactly $m_v$ \emph{special}, and
condition also on a \emph{masked} cycle length process $\vec\ell$. The process $\vec\ell$ hides the special cycles, but reveals all other cycle lengths in order, until stopping once a moderately large fixed $\poly(n)$ number of non-special cycles have been revealed.\footnote{Technically, we will assign length $m_v$ cycles to be non-special with some positive probability. This is necessary if the cycle length \emph{always} equals $m_v$, for example.}

As an explicit example, suppose that at vertex $v$, there is a $1/4$ probability of arriving back at $v$ after $3$ steps, $1/4$ probability of arriving back at $v$ after $5$ steps,  $1/3$ probability of arriving back at $v$ after $7$ steps, and $1/6$ probability of arriving back at $v$ after $11$ steps, and let $m_v = 7$. Suppose that on a specific run of the Markov chain starting at $v$, the sequence of cycles taken are of lengths $(3, 5, 7, 11, 7, 5, 3, 7)$. Then $\vec\ell = (3, 5, 11, 5, 3)$.

After this point, the number of hidden $m_v$ cycles is approximately Gaussian with $\poly(n)$ standard deviation.
In particular, its probability mass function is almost constant on short scales,
which lets us approximate the function
\[
    \bbP[x_s\in U~|~x_t=v]
\]
by an $m_v$-periodic function for large $s$. 
Averaging over the randomness of $x_t$, we have approximated $\bbP[x_s\in U]$ by an average of periodic functions with periods $m\leq 10n$.\footnote{We formalize this approximation using a notion of comparison we call $c$-covering, see Definition~\ref{def:cover}.}
This completes the structural phase of the proof. 
We now turn to finding a time $s$ such that $\bbP[x_s\in U]\in [1/3,2/3]$.

\paragraph{Analysis of Periodic Decomposition.}

The last step is to analyze this mixture-of-cyclic behavior at a time-region $T$ on which $\bbP[x_s\in U]\geq 2/3$ holds for roughly half of the values $s\in [T,T+\poly(n)]$; such $T$ exists by the discrete-time intermediate value theorem.
If the different periods $m_v$ were relatively prime, intuitively each pair of cyclic behaviors would be approximately independent over the range $[T,T+\poly(n)]$ since $m_v m_v'\leq 100n^2\ll \poly(n)$. This pairwise almost independence would allow us to upper bound $\sum_{s\in [T,T+\poly(n)]} (\bbP[x_s\in U]-1/2)^2$ similarly to Example $2$, thus contradicting pseudo-determinism.

To handle general period lengths, the key idea is to divide out common prime factors and restrict to a corresponding arithmetic progression. 
Precisely, for each prime $p$ we let $e_p\geq 0$ be the largest integer such that $m_v$ has at least some constant probability to be a multiple of $p^{e_p}$. 
Then it can be shown that $\beta=\prod_{p\text{ prime}} p^{e_p}\leq \poly(n)$, so we restrict attention to a fixed arithmetic progression with difference $\beta$.

Restricting to this arithmetic progression causes the different cycles to behave almost as if they are relatively prime. 
Precisely, the ``reduced period lengths'' $M_v=m_v/\gcd(m_v,\beta)$ have the property that any fixed prime $p$ has low probability of dividing $M_v$. 
This property does \emph{not} imply pairwise independence of different cycles: it could still be true that $\gcd(M_v,M_v')>1$ holds with high probability for independent $M_v,M_v'$.
However using Fourier analysis, we were still able to show that it implies a non-trivial upper bound on $\sum_{s\in [T,T+\poly(n)]} (\bbP[x_s\in U]-1/2)^2$. As mentioned above, such an upper bound contradicts pseudo-determinism which completes the proof.

\section{Proof of Theorem \ref{thm:main2}}
\label{sec:solution}

Assume $n\geq n_0$ is sufficiently large and let the threshold $N>n^{34}$. We will show such a Markov chain cannot $(N,M)$-solve the approximate counting problem for any $M$. We will replace the intervals $[0,1/3]$ and $[2/3,1]$ by $[0,10^{-4}]$ and $[1-10^{-4},1]$ in the statement of Theorem~\ref{thm:main2}, which is equivalent by amplification.

\subsection{Recurrent Behavior on Moderate Time-Scales}
\label{sec:cycles-def}

Given any $v\in V$, let $\mu_v$ be the \emph{return-time distribution} from $v$, i.e. the distribution for the smallest $r\geq 1$ such that $x_r=v$, starting from $x_0=v$. (We let $r=\infty$ if this never holds again.)

\begin{lemma}
    \label{lem:find-t}
    There exists $1\leq t\leq N/100$ and a subset $S\subseteq V$ such that:
    \begin{enumerate}[label=(\roman*)]
        \item $\bbP[x_t\in S]\geq 1/2$.
        \item\label{cond:return-prob} For all $v\in S$:
        \begin{align*}
            \bbP_{r \sim \mu_{v}}(r\in [0,10n])
            &\geq 
            1/2,
            \\
            \bbP_{r \sim \mu_{v}}(r\in [0,n^{18}])
            &\geq 
            1-n^{-16}.
        \end{align*}
    \end{enumerate}
\end{lemma}

\begin{proof}
    We first show that at least $0.85$-fraction of the values of $1\leq t\leq N/100$ satisfy: 
    \begin{equation}
    \label{eq:markov-short}
        \bbE_{x_t} \bbP_{r \sim \mu_{x_t}}(r\in [0,10n]) \geq 
            0.85
    \end{equation}
    Partition the interval $[1,N/100]$ into intervals $I_j=[1+(j-1)10n,10jn]$ of length $10n$ for $1\leq j\leq N/1000n$. Consider a fixed run of the Markov chain $x_1\ldots x_{N/100}$. For each interval $I_j$, at most $n$ values of $z_i$ for $i\in I_j$ will not return to themselves at some later time in $I_j$. Thus, in a fixed run of the Markov chain, the return time back to $x_t$ starting from time $t$ is at most $10n$ steps for at least $0.9$ fraction of times $1\leq t\leq N/100$. Averaging over all runs of the Markov chain, we get that 
    \[ 
        \bbE_{t\in [1,N/100]} \bbE_{x_t} \bbP_{r \sim \mu_{x_t}}(r\in [1,10n]) \geq 0.9
    \]
    and therefore at least $0.85$ fraction of $0<t<N/100$ satisfy \eqref{eq:markov-short}
    by Markov's inequality.
    By an identical argument, for at least $0.85$ fraction of $1\leq t\leq N/100$,
    \[ 
        \bbE_{x_t} \bbP_{r \sim \mu_{x_t}}(r\in [0,n^{18}]) \geq 
            1-10n^{-17}
    \]
    In particular, there exists a value of $t$ satisfying both conditions, which will be the value of $t$ in the lemma statement.
    
    Next applying Markov's inequality to  \eqref{eq:markov-short} shows for some $S_1\subseteq V$ with $\bbP[x_t\in S_1]\geq 0.6$,
    \[
        \bbP_{r \sim \mu_{v}}(r\in [1,10n])
        \geq 
        1/2.
    \]
    Similarly for some $S_2\subseteq V$ such that $\bbP[x_t\in S_2]>0.9$,
    \[
        \bbP_{r \sim \mu_{v}}(r\in [0,n^{18}])
        \geq 
        1-100n^{-17} \geq 1-n^{-16}.
    \]
    Taking $S=S_1 \cap S_2$ completes the proof.
\end{proof}

\subsection{Decomposition into Periodic Parts}

From now on we fix $t,S$ as in the previous subsection.
For each $v\in S$, define the period length 
\[
m_v \equiv \argmax_{1\leq m\leq 10n} \mu_v(m).
\]

Starting from $x_t=v$, we define a sequence of steps as follows. 
Let $(C_1,C_2,\dots)$ be the sequence of cycle lengths returning back to location $x_t$, i.e. $x_{t+C_1+\dots+C_j}$ for $j\geq 0$ are exactly the times $s\geq t$ with $x_s=x_t$. 
When $C_i=m_v$, we assign $C_i$ to be \emph{special} with probability 
\begin{equation}
\label{eq:special-prob}
    \bbP[C_i\text{ special}~|~C_i=m_v]
    =
    \frac{n^{-2}\mu_v(\{1,2,\dots,n^{18}\})}{\mu_v(m_v)}.    
\end{equation}
This choice of probability results in $\bbP[C_i\text{ special}~|~C_i\leq n^{18}]=n^{-2}$, which will be useful later. In all other cases where $C_i\leq n^{18}$, we assign $C_i$ to be \emph{typical}.
We define the \emph{special-steps sequence} $(s_1,s_2,\dots)=(m_v,m_v,\dots)$ to be the subsequence of special cycle lengths, and the \emph{typical-steps sequence} $(\ell_1,\ell_2,\dots)$ to consist of the typical cycle lengths.
Note that cycles of length greater than $n^{18}$ are neither special nor typical.

If $x_t=v$, let $\tau_v$ be the first time that $n^{14}$ typical steps have been completed. (We let $\tau_v$ be undefined or infinite for all other states $v$.)
Let $Y$ be the number of special cycles until time $\tau_{x_t}$, and let $s_{\tot}=Ym_v$ be the total duration of these special cycles.

For each $v\in S$, let $E_v$ be the event that $x_t=v$, and that all cycles until time $\tau_v$ are special or typical (in particular $\tau_v<\infty$ on this event). 

\begin{lemma}\label{lem:stop-infty}
    For any $v\in S$, we have $\bbP[E_v~|~x_t=v]\geq 0.99$.
\end{lemma}

\begin{proof}
    Note that each cycle length $C_i$ is independent and has probability at least $1-n^{-16}$ of being at most $n^{18}$ by Lemma~\ref{lem:find-t}.
    We find that with high probability the first $n^{15}$ cycle lengths are special or typical:
    \begin{equation}
    \label{eq:all-cycles-good}
    \bbP
    \lt[
    \max_{1\leq i\leq n^{15}} C_i\leq n^{17}
    \rt]
    \geq 1-n^{-1}.
    \end{equation}
    Conditioned on the event in \eqref{eq:all-cycles-good}, each $C_i$ is labeled typical independently with probability at least $1/2$.
    Hence the number of typical steps among the first $n^{15}$ cycles stochastically dominates a binomial $Bin(n^{15},1/2)$ variable, and in particular is at least $n^{14}$ with probability at least $1-n^{-1}$. 
    On this event, $\tau_v$ occurs during the first $n^{15}$ cycles and so $E_v$ holds with overall probability at least $(1-n^{-1})(1-n^{-1})\geq 0.99$.
\end{proof}

Next, we use our concept of special and typical cycles to describe the distribution of $\tau_v$, which we recall is the first time that $n^{14}$ typical steps have been completed conditioned on $x_t=v$.

\begin{definition}
\label{def:interval-periodic}
A probability distribution $\nu$ on $\bbZ$ is \textbf{$m$-interval-periodic} if it is supported on a discrete interval
$\{I,I+1,\dots,J\}$ and within this interval, $\nu(j)=\nu(j+m)$ depends only on $j\mod m$.
The \textbf{range} of $\nu$ is $J-I$.

Similarly a sequence $\vec\nu=(\nu_{I},\nu_{I+1},\dots, \nu_{J})$ of probability distributions on $V$ is $m$-interval-periodic if $\nu_j=\nu_{j+m}$ depends only on $j\mod m$. 
\end{definition}

\begin{definition}
\label{def:cover}
    Let $\mu,\nu$ be probability distributions on a countable set $\cI$ and $0<c<1$ be a constant. We say $\mu$ is a $c$-cover for $\nu$ if $\mu(i)\geq c\nu(i)$ for all $i\in \cI$.
\end{definition}

Let $\mu_{\tot,v,\vec\ell}$ be the distribution of $\tau_v$ conditionally on the event $E_v$ and the sequence $\vec\ell$ of $n^{14}$ typical steps.
In fact, it is easy to see from \eqref{eq:special-prob} that conditionally on $(v,E_v)$ the number $z_i$ of special steps between each adjacent pair $(\ell_{i-1},\ell_i)$ of typical steps is exactly the geometric distribution $\Geom(1-n^{-2})$. Thus $\bbP[z_i=j]=(1-n^{-2})n^{-2j}$ for each $j\in\bbZ_{\geq 0}$, and the mean and variance are known to be 
\begin{align*}
    \bbE[z_i]&=\frac{1}{n^2-1},
    \\
    \Var[z_i]&= \frac{n^2}{(n^2-1)^2}.
\end{align*}

\begin{lemma}\label{lem:short-steps-distribution}
    For each $v\in S$ and $\vec\ell$, there exists an $m_v$-interval-periodic distribution $\nu_{v,\vec\ell}$ with range $n^{6}<w_v<n^{8}$ such that the distribution $\mu_{tot,v,\vec\ell}$ is a $0.06$-cover of $\nu_{v,\vec\ell}$.
\end{lemma}

\begin{proof}
    Let $Z=\sum_{i=1}^{n^{14}} z_i$ be the total number of special cycles until time $\tau_v$.
    We first show that $Z$ obeys a central limit theorem as $n\to\infty$.
    Note that $w_i=z_i-\bbE[z_i]$ satisfies $\bbE[w_i]=0$, $\bbE[w_i^2]=\Theta(n^{-2})$ and $\bbE[|w_i|^3]=\Theta(n^{-2})$.
    Hence it follows from the Berry-Esseen theorem that with $\Phi(A)=\bbP^{y\sim \cN(0,1)}[y\leq A]$ the cumulative distribution function of the standard Gaussian, and $\mu_Z,\sigma_Z$ the mean and standard deviation of $Z$:
    \[
    \sup_{A\in\bbR}
    |\bbP\lt[\frac{Z-\mu_Z}{\sigma_Z}\leq A\rt]
    -\Phi(A)
    |
    \leq 
    \frac{Cn^{-2}}{n^{-3}\sqrt{n^{14}}}
    \leq O(1/n^2).
    \]
    In particular, for $n$ large enough,
    \begin{align*}
    \bbP[Z\in [\mu_Z-2\sigma_Z,\mu_Z-\sigma_Z]]
    &\geq \Phi(-1)-\Phi(-2)-0.001\geq  0.13,
    \\
    \bbP[Z\in [\mu_Z+\sigma_Z,\mu_Z+2\sigma_Z]]
    &\geq \Phi(2)-\Phi(1)-0.001 \geq 0.13.
\end{align*}
By averaging, there exist integers $a_1\in [\mu_Z-2\sigma_Z,\mu_Z-\sigma_Z]$ and $a_2\in [\mu_Z+\sigma_Z,\mu_Z+2\sigma_Z]$ with 
\[
    \bbP[Z=a_1],\bbP[Z=a_2]\geq \frac{1}{8\sigma_Z}.
\]
Next, we claim the probability mass function of $Z$ is log-concave; this implies that 
\begin{equation}
\label{eq:unimodal-trick}
    \bbP[Z=a]\geq \frac{1}{8\sigma_Z},\quad\forall a\in [\mu_Z-\sigma_Z,\mu_Z+\sigma_Z].
\end{equation}
Since $\mu_{\tot,v,\vec\ell}$ is exactly the law of $m_v Z$ and $\sigma_Z=n^{7}\sqrt{\Var[z_1]}=n^6(1\pm O(1/n))$, the statement \eqref{eq:unimodal-trick} readily implies the desired result. It remains only to prove that $Z$ has log-concave probability mass function.

For this, first note that $Z$ is a negative binomial distribution and its probability mass function takes the exact form
\[
    \bbP[Z=k]
    =
    \binom{k+n^{14}-1}{n^{14}-1} n^{-14k} (1-n^{-2})^{n^{14}}.
\]
(This description is well-known and follows from an elementary stars and bars argument.)
Hence it suffices to show that $k\mapsto \binom{k+n^{14}-1}{n^{14}-1}$ is log-concave. For this we note that $k\mapsto k+c$ is log-concave for each $c\geq 0$ on $k\in\bbZ_{\geq 0}$ and
\[
    \binom{k+n^{14}-1}{n^{14}-1}
    =
    \frac{(k+n^{14}-1)(k+n^{14}-2)\cdots(k+1)}{(n^{14}-1)!}
\]
is proportional to a product of such sequences.
\end{proof}

We remark that the conditioning on $\vec\ell$ played no role above as it just shifts $\mu_{tot,v,\vec\ell}$. However, in the next subsection it will be important to apply Lemma~\ref{lem:short-steps-distribution} after conditioning on $\vec\ell$. 

For each $v\in S$, typical-steps sequence $\vec\ell$, and time $s>t$, define the probability distribution
\[
\wt\mu_s=\wt\mu_s(v,\vec\ell)
\]
on $V$ to be the conditional law of $x_s$ given $(E_v,\vec \ell)$.

\begin{lemma}\label{lem:short-steps-distribution-v2}
    For each $v\in S$ and $\vec \ell$ and $T'>n^{33}$, there exists an $m_v$-interval-periodic sequence $\vec \omega=\vec\omega_{v}=(\omega_{T'},\ldots,\omega_{T'+n^{5}})$ of distributions that is $0.03$-covered by the sequence $\wt\mu_{T'}\ldots \wt\mu_{T'+n^{5}}$.
\end{lemma}

\begin{proof}
    For any vertex $u$ and time $r$, define 
    $P^r(u)$ to be the distribution of the Markov chain state after $r$ steps starting from vertex $u$.
    Let 
    \[
    \supp(\nu_{v,\vec\ell})=
    \{I,I+1,\dots,J\}=
    \{
    I(\nu_{v,\vec\ell})
    ,I(\nu_{v,\vec\ell})+1,
    \dots,J(\nu_{v,\vec\ell})\}
    \]
    be the support of $\nu_{v,\vec\ell}$. Moreover let $\overline{\nu}_{v,\vec\ell}(0),\dots,\overline{\nu}_{v,\vec\ell}(m_v-1)$ be such that $\nu_{v,\vec\ell}(j)=\overline{\nu}_{v,\vec\ell}(j\mod m)$ for $I\leq j\leq J$.

    We explicitly compute $\wt\mu_s$ for $T'\leq s\leq T'+n^{5}$. With inequalities of positive measures indicating that the difference is a positive measure, we have
    \begin{align*}
        \wt\mu_s(v,\vec\ell)
        &\geq 
        \sum_{d=1}^{s}
        \bbP[\tau_v=d|(x_d=v,\vec\ell)]\cdot P^{s-d}(v)
        \\
        &\geq 
        0.06
        \sum_{d=1}^{s}
        \nu_{v,\vec\ell}(d)\cdot P^{s-d}(v)
        \\
        &=
        0.06\sum_{u=I-T'}^{J-T'-n^{5}}
        \nu_{v,\vec\ell}(s-u)\cdot P^{u}(v)
        \\
        &\geq
        0.06\lt(
        \frac{J-I-n^{5}-2m}{m}
        \rt) 
        \cdot
        \sum_{i=1}^m
        \overline{\nu}_{v,\vec\ell}(s-i\mod m) 
        \sum_{\substack{I-T'\leq u\leq J-T'-n^{5}\\ u\equiv i\mod m}}
        P^u(v).
    \end{align*}
    Note that this lower bound is exactly $m$-periodic for $s$ in the stated range. We now show it has a significant total mass. Note that 
    \begin{align*}
    &\lt(
    \frac{J-I+n^{5}+2m}{m}
    \rt) 
    \cdot
    \sum_{i=1}^m
    \overline{\nu}_{v,\vec\ell}(s-i\mod m) 
    \sum_{\substack{I-T'\leq u\leq J-T'-n^{5}\\ u\equiv i\mod m}}
    P^u(v)
    \\
    &\geq 
    \sum_{t=1}^{s}
    \nu_{v,\vec\ell}(t)\cdot P^{s-t}(v)
    \end{align*}
     has total mass at least $0.99$ since clearly $\bbP[\tau_v<s~|~E_v]\geq 0.99$ (by Markov's inequality on $s_{tot}$). 
    
    Since $J-I\geq n^{6}$ we have
    \[
    \frac{J-I-n^{5}-2m}{J-I+n^{5}+2m}\geq 0.99
    \]
    so we conclude that the lower bound above has a total mass of at least $0.03$.
\end{proof}

\subsection{Analysis of Periodic Decomposition}

We fix the values $m_v$ and distributions $\nu_{v,\vec\ell}$ from Lemma~\ref{lem:short-steps-distribution}. 
Next, we construct a ``global period length'' $\beta$ as follows. For each prime $p$, let $e_p$ be the largest value such that 
\[
\bbP[p^{e_p}\text{ divides }m_v~|~v\in S]\geq 0.55
\]
and define $\beta=\prod_p p^{e_p}$.

\begin{lemma}
    It holds that $\beta\leq n^{2}$.
\end{lemma}

\begin{proof}
    It suffices to observe that
    \[
        \log(10n)
        \geq
        \bbE[\log(m_v)~|~v\in S]
        \geq 
        \log(\beta^{0.55}).
    \]
    The former holds since $m_v\leq 10n$ while the latter holds by Markov's inequality. Hence $\beta\leq (10n)^{20/11}\leq n^2$.
\end{proof}

Let $M_v=m_v/\gcd(m_v,\beta)$. By definition, for any prime $p$ we have 
\begin{equation}
\label{eq:few-primes-p}
    \bbP[p\text{ divides }M_{x_t}~|~x_t\in S]\leq 0.55.
\end{equation}
Recalling Subsection~\ref{subsec:markov-chain-setup},  for each time $s$
let $F(s)=0$ if $\bbP[x_s\in U]\in [0,1/3]$ and $F(s)=1$ otherwise. Moreover define
\[
    \overline{F}_a=
    \frac{1}{\lfloor n^{5}/\beta \rfloor}
    \sum_{i=1}^{\lfloor n^{5}/\beta \rfloor}
    F((a+i)\beta)
    .
\]

\begin{lemma}
\label{lem:IVT}
    There exists $a\in \bbN$ such that $|\overline{F}_a-1/2|\leq 0.01$.
\end{lemma}

\begin{proof}
    This follows by the discrete-time intermediate value theorem. Indeed since $N\geq n^{5}$ we have $\overline{F}_0=1$, while $\overline{F}_a=0$ for $a\geq f(N)$. Clearly $|\overline{F}_{a+1}-\overline{F}_a|\leq 0.01$ for all $a$.
\end{proof}

Throughout the rest of the proof, let $T=a\beta$ for $a$ as in Lemma~\ref{lem:IVT}. We note that $T>n^{33}$ and therefore satisfies the requirements for $T'$ in Lemma~\ref{lem:short-steps-distribution-v2}.
For each time $T\leq s\leq T+n^{5}$, let 
\[
    \overline{\omega}_s=
    \bbE^{x_t,\vec\ell}
    \lt[
    \bbP^{v\sim \omega_{x_t,\ell,s}}
    [v\in U~]~|~x_t\in S \text{ and } E_{x_t}
    \rt].
\]
(Note that $\{x_t\in S \text{ and } E_{x_t}\}$ is just an event.)

\begin{lemma}
\label{lem:omega-biased}
    $\overline{\omega}_s\geq 0.99$ if $F(s)=1$ and 
    $\overline{\omega}_s\leq 0.01$ if $F(s)=0$.
\end{lemma}

\begin{proof}
    We focus on the latter statement since they are symmetric.

    Given $x_t$, generate $w\sim P^{s-t}(x_t)$ to be the state of the Markov chain at time $s$ drawn from the original trajectory distribution.
    By Bayes' rule for expectations, since the event $\{x_t\in S\text{ and }E_{x_t}\}$ has probability at least $1/3$, conditioning on it increases the expectation of any non-negative random variable by a factor at most $3$. 
    In particular:
    \begin{align*}
    \bbE^{x_t,\vec\ell}
    \lt[
    \bbP
    [w\in U~]~|~x_t\in S\text{ and }E_{x_t}
    \rt]
    &\leq 
    \frac{\bbE^{x_t,\vec\ell}
    \lt[
    \bbP
    [w\in U~]\rt]}{\bbP[x_t\in S\text{ and }E_{x_t}]}
    \\
    &=
    \frac{\bbP[x_s\in U]}{\bbP[x_t\in S\text{ and }E_{x_t}]}
    \\
    &\leq 
    3\cdot 10^{-4}.
    \end{align*}
    Given $x_t$ and $\vec\ell$, we also let $\wt w\sim \omega_{x_t,\vec\ell,s}$ for $\omega_{x_t,\vec\ell,s}$ as in Lemma~\ref{lem:short-steps-distribution-v2}.
    Then the $0.03$-covering guarantee in Lemma~\ref{lem:short-steps-distribution-v2} directly implies
    \begin{align*}
    &\bbE^{x_t,\vec\ell}
    \lt[
    \bbP
    [w\in U~]~|~x_t\in S\text{ and }E_{x_t}
    \rt]
    \\
    &\geq 
    0.03
    \cdot 
    \bbE^{x_t,\vec\ell}
    \lt[
    \bbP
    [\wt w\in U~]~|~x_t\in S\text{ and }E_{x_t}
    \rt]
    \\
    &=
    0.03\cdot \overline{\omega}_s
    .
    \end{align*}
    Combining completes the proof.
\end{proof}

\begin{lemma}
It holds that
\label{lem:periodic-unbiased}
    \[
    \lt|\frac{1}{\lfloor n^{5}/\beta \rfloor}
    \sum_{i=1}^{\lfloor n^{5}/\beta \rfloor}
    \overline{\omega}_{(a+i)\beta}
    -\frac{1}{2}\rt|
    \leq 
    0.02.
    \]
\end{lemma}

\begin{proof}
    By Lemma~\ref{lem:omega-biased}, 
    $|\overline{\omega}_{s}-F(s)|\leq 0.01$ for each $s$. Hence
    \[
    \lt|\frac{1}{\lfloor n^{5}/\beta \rfloor}
    \sum_{i=1}^{\lfloor n^{5}/\beta \rfloor}
    \overline{\omega}_{(a+i)\beta}
    -\frac{1}{2}\rt|
    \leq 
    |\overline{F}_a-1/2|+0.01\leq 0.02.
    \]
\end{proof}

We will use the following lemma which is proved in Section~\ref{sec:fourier-proof}.

\begin{lemma}
\label{lem:fourier}
The following holds with $L=n^3$. Fix $n$ and let $S=(S_1,S_2,\dots,S_{L})\in [-1,1]^L$ be a random sequence which is $m$-periodic for some random $m=m(S)\leq 10n$. 
Suppose that 
\begin{equation}
\label{eq:low-bias-assumption}
    -L/10\leq \bbE\sum_{d=1}^L S_d\leq L/10
\end{equation}
and that for each prime $p$,
\[
    \bbP[p\text{ divides }m]\leq 0.55.
\]
Then there exists an index $1\leq d \leq L$ such that 
\[
    |\bbE[S_d]|\leq 0.9.
\]
\end{lemma}

Combining the above, we can finally prove Theorem~\ref{thm:main2}.

\begin{proof}[\textbf{Proof of Theorem~\ref{thm:main2}}:]
Given a putative $(\cM_n,U)$, we apply Lemma~\ref{lem:fourier} with 
\[
    S=S_{x_d,\vec\ell}=\bbP^{v\sim \omega_{x_d,\ell,s}}
    [v\in U].
\]
where we map $[0,1]\to [-1,1]$ via $a\mapsto 2a-1$.
The first condition holds by Lemma~\ref{lem:periodic-unbiased} since $\bbE[S_d]=\overline{\omega}_{(a+d)\beta}$.
The second condition holds by \eqref{eq:few-primes-p}. 
Hence Lemma~\ref{lem:fourier} implies that there exists $s$ such that $\overline{\omega}_s \in [0.02,0.98]$. This contradicts Lemma~\ref{lem:omega-biased} above, giving a contradiction and completing the proof. 
\end{proof}

\section{Proof of Lemma~\ref{lem:fourier}}
\label{sec:fourier-proof}

Here we prove Lemma~\ref{lem:fourier} using Fourier analysis. 
For each $m\geq 1$, let 
\[
Z_m=\{0,1/m,\dots,(m-1)/m\}
\]
and set $Z_m^*=Z_m\backslash \{0\}$.
Given a $m$-periodic sequence $S=(S_1,S_2,\dots)$ let $\wh S:[0,1)\to \bbC$
be its Fourier transform
\begin{equation}
\label{eq:fourier-def}
    \wh S(\alpha)
    =
    \begin{cases}
    \frac{1}{m}
    \sum_{j=1}^{m}
    e^{2\pi i j\alpha} S_j,\quad\alpha\in Z_m,
    \\
    0,\quad\quad
    \quad\quad\quad\quad\quad\alpha\notin Z_m\,.
    \end{cases}
\end{equation}
The definition \eqref{eq:fourier-def} agrees with the Fourier transform of $S$ viewed as a periodic function on $\bbZ$ and hence is independent of the period $m$ --- we could view $S$ as being $km$ periodic for any $k\geq 1$ and $\wh S$ would remain consistent. This allows us to use Fourier analysis on pairs $S,S'$ of sequences with different period lengths $m,m'$. 
Indeed with $M=mm'$ we can define:
\begin{align}
\label{eq:inner-product}
    \la S,S'\ra
    &=
    \frac{1}{M}
    \sum_{i=1}^{M}
    S_i S'_i
    \\
\label{eq:primal-norm}
    \|S\|_{L^2}^2
    &=
    \frac{1}{M}
    \sum_{i=1}^{M}
    S_i^2
    \\
\label{eq:fourier-inner-product}
    \la \wh S,\wh S'\ra
    &=
    \sum_{\alpha\in Z_M}
    \wh S(\alpha) \wh S'(\alpha)
    \\
\label{eq:dual-norm}
    \|\wh S\|_{\ell^2}^2
    &=
    \sum_{\alpha\in Z_m}
    |\wh S(\alpha)|^2
    .
\end{align}
Parseval identity's modulo $M$ implies that $\la S,S'\ra=\la \wh S,\wh S'\ra$ and $\|S\|_{L^2}=\|\wh S\|_{L^2}$.

The next proposition approximates averages on $1,2,\dots,L$ by Fourier averages.
Here and below we let $\mu$ be the law of the random sequence $S$, and recall that $S\sim\mu$ is $m(S)$ periodic for $m\leq 10n$ almost surely.

\begin{proposition}
\label{prop:periodic-approximation}
    Then
    \begin{align}
    \label{eq:inner-product-periodic-approximation}
    \lt|
    \frac{1}{L}
    \sum_{d=1}^L
    \bbP^{S\sim \mu}[S_d=1]^2
    -
    \bbE^{S,S'\sim\mu}
    \la \wh S,\wh S'\ra
    \rt|
    &\leq 
    O(n^2/L),
    \\
    \label{eq:bias-periodic-approximation}
    \lt|
    \lt(
    \frac{1}{L}
    \bbE^{S\sim\mu}
    \sum_{d=1}^L
    S_d\rt)^2
    -
    \bbE^{S,S'\sim\mu}
    [\wh S(0)\wh S'(0)]
    \rt|
    &\leq 
    O(n^2/L).
    \end{align}
\end{proposition}

\begin{proof}
    Notice that by definition,
    \[
    \frac{1}{L}
    \sum_{d=1}^L
    \bbP^{S\sim \mu}[S_d=1]^2
    =
    \frac{1}{L}\bbE^{S,S'\sim\mu}
    \big\la (S_1\dots,S_L),(S'_1,\dots,S'_T)\big\ra
    \]
    where $\la (S_1\dots,S_L),(S'_1,\dots,S'_T)\ra=\frac{1}{L}\sum_{d=1}^L S_d S'_d$.
    Since $M=mm'\leq 100n^2$,
    \[
    \lt|
     \frac{1}{L}\bbE^{S,S'\sim\mu}
    \big\la (S_1\dots,S_L),(S'_1,\dots,S'_T)\big\ra
    -
     \la S,S'\ra
    \rt|
    \leq 
    O(n^2/L).
    \]
    Combining and using $\la S,S'\ra=\la \wh S,\wh S'\ra$ by Parseval completes the proof of \eqref{eq:inner-product-periodic-approximation}. The proof of \eqref{eq:bias-periodic-approximation} is similar since 
    \[
    \lt|
    \frac{1}{L}
    \bbE^{S\sim\mu}
    \sum_{d=1}^L S_d
    -
    \bbE^{S\sim\mu}
    [\hat S(0)]
    \rt|
    =
    \lt|
    \frac{1}{L}\bbE^{S\sim\mu}
    \sum_{d=1}^L S_d
    -
    \frac{1}{m(S)}\bbE^{S\sim\mu}
    \sum_{i=1}^{m(S)} S_i
    \rt|
    \leq 
    O(n/L).
    \]
    Indeed $(\bbE^{S\sim\mu}[\wh S(0)])^2=\bbE^{S,S'\sim\mu}[\wh S(0)\wh S'(0)]$, and the function $x\mapsto x^2$ is Lipschitz on $x\in [-1,1]$.
\end{proof}

The main idea to prove Lemma~\ref{lem:fourier} is that $S,S'$ are unlikely to have similar periods. In carrying out this argument we have to handle the bias $\wh S(0)$ separately from the $\alpha\in (0,2\pi)$ contributions. 
The next proposition gives an estimate to handle the latter contributions; note it is important that the last term below mixes $S,S'$.

\begin{proposition}
\label{prop:gcd-fourier-bound}
    For $m$ and $m'$-periodic functions $S,S':\bbZ\to [-1,1]$, 
    \[
    \la \wh S,\wh S'\ra 
    \leq 
    \wh S(0)\wh S'(0)
    +
    \|\, \wh S|_{Z_{m'}^*}\|_{\ell^2}
    \]
    where $\wh S|_{Z_{m'}^*}$ denotes the restriction of $\wh S$ to $Z_{m'}^*$.
\end{proposition}

\begin{proof}
    The summand $\wh S(\alpha)\wh S'(\alpha)$ in \eqref{eq:fourier-inner-product} is non-zero only when $\alpha\in Z_g$ for $g=\gcd(m,m')$.
    This yields an upper bound of
    \[
    \wh S(0)\wh S'(0)
    +
    \|\,\wh S|_{Z_g^*}\|_{\ell^2}
    \cdot
    \|\,\wh S'|_{Z_g^*}\|_{\ell^2}.
    \]
    This immediately implies the claim since 
    \[
    \|\,\wh S|_{Z_g^*}\|_{\ell^2}
    \leq 
    \|\, \wh S|_{Z_{m'}^*}\|_{\ell^2}
    \]
    (in fact, equality holds)
    while 
    \[
    \|\,\wh S'|_{Z_g^*}\|_{\ell^2} 
    \leq 
    \|\,\wh S'|_{Z_M}\|_{\ell^2} 
    =
    \|\,S'\|_{L^2} 
    \leq 
    1.
    \]
\end{proof}

\begin{lemma}
\label{lem:small-non-biased-terms}
    Suppose $S,S'\stackrel{\iid}{\sim}\mu$ and $\bbP^{S\sim\mu}[p\text{ divides }m(S)]\leq \eps$ for any prime $p$. Then 
    \[
    \bbE^{S\sim\mu}
    \big[
    \|\, \wh S|_{Z_{m'}^*}\|_{\ell^2}
    \big]
    \leq 
    \sqrt{\eps}
    \]
\end{lemma}

\begin{proof}
    Directly by the assumption, for any fixed $\alpha\neq 0$, we have $\bbP[\alpha\in Z_{m'}^*]\leq \eps$.    
    Thus linearity of expectation implies
    \[
    \bbE^{S\sim\mu}
    \big[
    \|\, \wh S|_{Z_{m'}^*}\|_{\ell^2}^2
    \big]
    \leq 
    \eps.
    \]
    Applying Jensen's inequality completes the proof.
\end{proof}

\begin{proof}[\textbf{Proof of Lemma~\ref{lem:fourier}}:] 
    By applying Markov's inequality to \eqref{eq:inner-product-periodic-approximation} and noting $n^2/L\leq 1/n$ it suffices to show that 
    \begin{equation}
    \label{eq:fourier-finish}
    \bbE^{S,S'\sim\mu}
    \la \wh S,\wh S'\ra
    \stackrel{?}{\leq} 0.8 < 0.9^2.
    \end{equation}
    We first deal with the bias term of the left-hand side.
    Combining \eqref{eq:bias-periodic-approximation} and \eqref{eq:low-bias-assumption} yields
    \begin{align*}
    \bbE^{S,S'\sim\mu}[\wh S(0)\wh S'(0)]\leq 0.02.
    \end{align*}
    For the remaining terms $\alpha\neq 0$, combining Proposition~\ref{prop:gcd-fourier-bound} with Lemma~\ref{lem:small-non-biased-terms} (which holds with $\eps=0.55$ by assumption) yields
    \[
    \bbE^{S,S'\sim\mu}
    \sum_{\alpha\in Z_M^*}
    \wh S(\alpha)\wh S'(\alpha)
    \leq \sqrt{0.55}\leq 0.75.
    \]
    Summing the contributions establishes \eqref{eq:fourier-finish} and thus finishes the proof.
\end{proof}

\section*{Acknowledgements}

Thanks to Yang Liu for helpful discussions, and to Rachel Zhang, Naren Manoj, and Themistoklis Haris for comments on the manuscript.
Ofer Grossman was supported by the NSF GRFP as a graduate student at MIT, and was visiting the Simons Institute for the Theory of Computing while part of this work was completed. Meghal Gupta was supported by a Chancellor's Fellowship as a graduate student at UC Berkeley.

\small
\bibliographystyle{alpha}
\bibliography{bib}

\end{document}